\documentclass[12pt,onecolumn,draftclsnofoot]{IEEEtran}
\ifCLASSINFOpdf
\else
\fi
\usepackage{times,amsmath,epsfig,latexsym,multirow,array,amssymb,graphicx,multicol,stfloats,xcolor,textcomp,mathrsfs}
\usepackage[amsmath,amsthm,thmmarks]{ntheorem}

\usepackage{tikz,pgfplots}

\usepackage{bibentry}

\usepackage{caption2}

\newcounter{cdefinition}
\newtheorem{definition}[cdefinition]{Definition}
\newcounter{clemma}
\newtheorem{lemma}[clemma]{Lemma}
\newcounter{ctheorem}
\newtheorem{theorem}[ctheorem]{Theorem}
\newcounter{ccorollary}
\newtheorem{corollary}[ccorollary]{Corollary}

\usepackage{algorithm}
\usepackage{algpseudocode}
\usepackage{balance}

\usepackage{url}
\usepackage{hyperref}

\begin{document}
%
\title{Interference Alignment with Quantized Grassmannian Feedback in the K-user Constant MIMO Interference Channel}

\author{\IEEEauthorblockN{Mohsen Rezaee\\}
\IEEEauthorblockA{
Universit\"at Paderborn, Dept. of Electrical Engineering and Information Technology, Paderborn, Germany, \\ {\texttt{mohsen.rezaee@sst.upb.de}}\\} 

\and

\IEEEauthorblockN{Maxime Guillaud\\}
\IEEEauthorblockA{Huawei Technologies, France Research Center,\\ Mathematical and Algorithmic Sciences Lab, Paris, France\\
{\texttt{maxime.m.guillaud@ieee.org}}}
}


%


\maketitle

\begin{IEEEkeywords}K-user MIMO interference channel, interference alignment, CSI quantization, Grassmann manifold, limited CSI.
\end{IEEEkeywords}

\pagebreak

\begin{abstract}
A simple channel state information (CSI) feedback scheme is proposed for interference alignment (IA) over the $K$-user constant Multiple-Input-Multiple-Output Interference Channel (MIMO IC). The proposed technique relies on the identification of invariants in the IA equations, which enables the reformulation of the CSI quantization problem as a single quantization on the Grassmann manifold at each receiver.
The scaling of the number of feedback bits with the transmit power sufficient to preserve the multiplexing gain that can be achieved under perfect CSI is established. We show that the CSI feedback requirements of the proposed technique are better (lower) than what is required when using previously published methods, for system dimensions (number of users and antennas) of practical interest.
Furthermore, we show through simulations that this advantage persists at low SNR, in the sense that the proposed technique yields a higher sum-rate performance for a given number of feedback bits.
Finally, to complement our analysis, we introduce a statistical model that faithfully captures the properties of the quantization error obtained for random vector quantization (RVQ) on the Grassmann manifold for large codebooks; this enables the numerical  (Monte-Carlo) analysis of general Grassmannian RVQ schemes for codebook sizes that would be impractically large to simulate.
\footnote{Part of the results presented in this paper have appeared in M.~Rezaee and M.~Guillaud, ``Limited Feedback for Interference Alignment in the K-user MIMO Interference Channel,'' \emph{Proc. IEEE Information Theory Workshop (ITW)}, September 2012. The statistical model for the quantization error of RVQ (Section~\ref{GMperturbation}), has appeared together with an extension of the present results to the problem of CSI exchange on the backhaul -- not covered in the present paper -- in M.~Rezaee, M.~Guillaud, and F.~Lindqvist, ``CSIT Sharing over Finite Capacity Backhaul for Spatial Interference Alignment,'' Proc. International Symposium on Information Theory (ISIT), Jun. 2013. This work was performed while both authors were with the Institute of Telecommunications of Vienna University of Technology, Vienna, Austria. }
\end{abstract}



%
\IEEEpeerreviewmaketitle

\section{Introduction}
Multiple-antenna transceivers are known to improve the performance of wireless communication links compared to single-antenna systems. The increasing demand for high throughput and reliable transmission necessitates efficient use of Multiple-Input-Multiple-Output (MIMO) systems. In particular in multi-user networks where interference is a major concern, the availability of channel state information (CSI) at the transmitter is crucial in order to fully exploit the performance improvement of MIMO systems. 
In scenarios where the channel is not reciprocal (such as frequency-division duplex systems), the CSI has to be quantized and fed back to the transmitter. 
The mismatch between the true channel and the quantized channel results in a degradation in performance. 

In this article, we focus on interference alignment (IA) applied to the $K$-user constant MIMO IC. IA has been shown to achieve the optimal multiplexing gain (also called the degrees of freedom, DoF) over the $K$-user interference channel when perfect CSI is available at the transmitters \cite{Kuser:CJ}; it was introduced for the $K$-user MIMO IC in \cite{Gou_Jafar_DoF_MIMO_Kuser_IC_IT2010}. It consists in designing the precoders such that the total interference at each receiver lies in a space with minimum dimensions so that the remaining dimensions can be used for interference-free decoding. When only imperfect CSI is available, 
the channel mismatch not only reduces the effective channel gain but also causes interference between users. The performance of IA with imperfect CSI has been analyzed e.g. in \cite{ImPerCh:Guill,Xie_etal_IA_CSI_uncertainty_TCom13}.

Extensive research has been made on limited feedback schemes for point-to-point MIMO systems \cite[and references therein]{loveg}. In \cite{love}, codebook design is investigated when the receiver selects the best unitary precoder from a finite codebook and feeds back the index of the selected precoder to the transmitter. \cite{love} shows that the optimal design for such a codebook is equivalent to the Grassmannian subspace packing problem. Some useful quantization bounds on the Grassmann manifold are derived in \cite{dai,Barg}. In \cite{Wiro}, quantization of the precoding matrix using random vector quantization (RVQ) codebooks is investigated, providing insights on the asymptotic optimality of RVQ.

Concerning multi-user systems, the question of the scaling of the size of the codebook used for CSI feedback with increasing signal-to-noise ratio (SNR) has been explored in a number of recent works. 
Generally speaking, using imperfect CSI at the transmitter (CSIT) to compute the transmit precoders in a multi-user system causes interference at the receiver side. Since the power of this interference scales with the transmit power, it is necessary to compensate any increase in transmit power by decreasing the quantization error affecting the CSIT, if the interference at the receiver is to remain bounded. This has led several authors to study how the codebook size should scale with the SNR in order to preserve the degrees of freedom achievable with perfect CSI, for several feedback schemes. 
The case of the broadcast channel was considered first; assuming zero-forcing precoding and single-antenna receivers, it has been determined in \cite{Jindal_BC_finite_rate_feedback_IT06} that scaling the amount of feedback bits with $(M-1) \log P$ (where $M$ is the number of antennas at the transmitter and $P$ the transmit power) at each receiver is sufficient to achieve full DoF.
For the $K$-user IC, most results on CSI quantization focus on transmission schemes based on IA, since IA is instrumental in achieving the channel DoF \cite{Cadambeit08,Gou_Jafar_DoF_MIMO_Kuser_IC_IT2010}. Specifically, in that context, the CSI feedback problem is considered for $L$-tap frequency selective SISO links  in \cite{Thukral_Boelcskei_IA_limited_feedback_ISIT2009}, where it is shown that the channel DoF is achievable if the number of bits used to encode the CSI scales with $K(L-1)\log P$. This result was further extended to the $N\times M$ MIMO frequency-selective IC in \cite{rajac}, where $\min\{M,N\}^2 K(R L-1)\log P$ bits (with $R=\lfloor \frac{ \max\{M,N\} }{ \min\{M,N\} }\rfloor$) are shown to be required to achieve the perfect-CSI DoF. However, both \cite{Thukral_Boelcskei_IA_limited_feedback_ISIT2009} and \cite{rajac} rely on the same analysis, which is not applicable to the flat-fading case\footnote{It is noted in \cite{Thukral_Boelcskei_IA_limited_feedback_ISIT2009} that the result does not hold for low values of $L$, however the minimum $L$ for which it holds can not be conclusively ascertained from the article. We note that in particular, for the flat-fading case ($L=1$) of interest in this paper, both  \cite{Thukral_Boelcskei_IA_limited_feedback_ISIT2009} for the SISO case and \cite{rajac} for the MIMO square case ($M=N$) yield a scaling independent of $\log P$, which is unrealistic.}.

In \cite{KimMoonLeeLee_CSI_quantization_MIMOIC_TWC2012}, the authors introduce two quantization schemes for the MIMO flat-fading $K$-user IC. The first one is based on quantization on the composite Grassmann manifold (inspired by \cite{rajac}). The second method improves the quantization accuracy by introducing a virtual receive filter at each receiver which leaves the IA equations invariant; the quantization error can be reduced by optimizing this virtual filter, however the process is computationally complex and must be repeated for each codeword and each channel realization. No asymptotic (high SNR) analysis  is provided in \cite{KimMoonLeeLee_CSI_quantization_MIMOIC_TWC2012}; it is easy to figure out that the first considered method requires a scaling of $(K-1)(MN-1)\log P$ to achieve the channel DoF, however the scaling required for the second method to achieve full DoF is not clear.\\

In this paper, we present a novel CSI quantization and feedback scheme for IA over the K-user constant MIMO IC. The salient points of our contribution are:
\begin{itemize}
\item The proposed feedback scheme exploits the invariances in the IA equations to reduce the dimension of the quantization space, without requiring the heavy iterative processing of e.g. \cite{KimMoonLeeLee_CSI_quantization_MIMOIC_TWC2012}.
\item We characterize the scaling (with SNR) of the codebook size under which the proposed feedback scheme achieves the same DoF as with perfect CSIT. This scaling is shown to be better (slower) than the scaling obtained using the schemes from \cite{Jindal_BC_finite_rate_feedback_IT06} or \cite{KimMoonLeeLee_CSI_quantization_MIMOIC_TWC2012} for all system dimensions where IA is feasible.
\item At non-asymptotic SNR and for a fixed codebook size, the proposed scheme is shown by simulation to achieve better sum-rate performance than the methods from \cite{Jindal_BC_finite_rate_feedback_IT06} or \cite{KimMoonLeeLee_CSI_quantization_MIMOIC_TWC2012}.
\item As a by-product of our analysis, we introduce a statistical model that faithfully captures the properties of the quantization error of RVQ on the Grassmann manifold for large codebooks; we use it to generate rotations that closely approximate the true quantization error of RVQ. This tool enables numerical analysis of general Grassmannian RVQ schemes for large codebook sizes, without requiring the generation of the codebook nor the exhaustive search normally associated with the quantizer.
\end{itemize}


The remainder of the paper is organized as follows. In Section \ref{sec:model}, the system model is described. A reformulation of the CSI representation for the interference alignment problem is provided in Section \ref{sec:IA}. The limited feedback (quantized) scheme is presented in Section \ref{sec:LF}, while the achievable rates and DoF are analyzed in Section \ref{sec:rateanal}. Simulation results are presented in Section \ref{sec:simulation} together with the statistical RVQ error model, and conclusions are drawn in Section \ref{sec:conclusion}.\\

{\it Notation:} Non-bold letters represent scalar quantities, boldface lowercase and uppercase letters indicate vectors and matrices, respectively. ${\bf I}_N$ is the $N\times N$ identity matrix, while $\bf 0$ denotes an all-zeros matrix. The trace, conjugate, transpose, Hermitian transpose of a matrix or vector are denoted by ${\rm tr}(\cdot),(\cdot)^*, (\cdot)^{\rm T}, (\cdot)^{\rm H}$ respectively. The expectation operator over variable $X$ is represented by ${\rm E_{X}(\cdot)}$. The determinant of a matrix (or absolute value of a scalar) is represented by $|\cdot |$. ${\mathcal G}_{n,d}$ denotes the complex Grassmann manifold of dimensions $(n,d)$, i.e. the set of all $d$-dimensional vector subspaces of an $n$-dimensional vector space over $\mathbb{C}$.
The Frobenius norm of a matrix is denoted by $||\cdot ||_{\rm F}$ while the two-norm (spectral norm) of a matrix is represented by $||\cdot ||_2$. A block diagonal matrix is denoted by ${\rm Bdiag}(\cdot)$ with the argument blocks on its diagonal. $\mathcal{N}(0,1)$ (resp. $\mathcal{CN}(0,1)$) denotes the real (resp. circularly symmetric complex) Gaussian  distribution with zero mean and unit variance. The largest eigenvalue of a matrix is denoted by $\lambda_{\rm max} (\cdot)$. Finally, $\rm log $ represents the logarithm in base 2.\\


\section{System Model}
\label{sec:model}

A MIMO interference channel is considered in which $K$ transmitters communicate with their respective receivers over a shared medium.  For the sake of notational simplicity, we consider the symmetric case where each transmitter has $M$ antennas while each receiver is equipped with $N$ antennas, although the method discussed here applies to non-symmetric settings as well.

Assume that transmitter $j$ employs a precoding matrix ${{\bf V}}_j$ to transmit $d$ data streams to its respective receiver. The $N$-dimensional signal at receiver $i$ reads
\begin{equation}\label{E1}
	{{\bf{y}}_i} = {{{{\bf H}}}_{ii}}{{{\bf V}}_i}{{\bf{x}}_i} + \sum_{ \substack{1\leq j\leq K\\ j \ne i}}{{{{{\bf H}}}_{ij}}{{{{\bf V}}}_j}{{\bf{x}}_j}}  + {{\bf{n}}_i}
\end{equation}
in which ${{{{\bf H}}}_{ij}} \in {\mathbb{C}^{N \times M}}$ is the channel matrix between transmitter $j$ and receiver $i$, ${{{{\bf V}}}_{j}} \in {\mathbb{C}^{M \times d}}$ is a truncated unitary matrix (${\bf V}_j^{\rm H}{\bf V}_j={\bf I}_d$), and ${{\bf{x}}_{j}} \in {\mathbb{C}^{d }}$ is the symbol vector of transmitter $j$. Furthermore, ${{\bf{n}}_{i}} \in {\mathbb{C}^{N}}$ is the additive noise at receiver $i$ whose elements are distributed independently as $\mathcal{CN}(0,1)$. We assume Gaussian circularly symmetric i.i.d. signaling with ${\rm {E}}\left[  {{\bf{x}}_j}{\bf{x}}_j^{\rm H} \right] = {\frac{P}{d}}{\bf I}_{d},{\rm{ }}\,\,\,j = 1, \ldots ,K$, where $P$ denotes the per-user transmit power. 
Following \cite{Yetis_Gou_Jafar_Kayran_feasibility_of_IA09}, we assume that the channel coefficients are generic; in particular, this condition is fulfilled by any channel model where the coefficients are drawn independently from a continuous distribution, such as the classical Gaussian i.i.d. model.\\

\section{Proposed Grassmannian Feedback Scheme for Interference Alignment}
\label{sec:IA}

Let us consider the interference alignment problem of \cite{Gou_Jafar_DoF_MIMO_Kuser_IC_IT2010}, and assume that the CSI is fed back from the receivers to the transmitters\footnote{The underlying assumption here is that all $K$ transmitters can exchange CSI instantaneously and ``for free.'' Alternatively, one can consider a central node (to which all the CSI would be forwarded) where the precoders are computed and subsequently distributed to the transmitters; this distinction is immaterial, and the results presented here apply to both cases. Variations on these assumptions are considered in \cite{Rezaee_Guillaud_Lindqvist_clustered_IA_ISIT2013}.}. 
Specifically, assume that the $i$th receiver has perfect knowledge of the channel matrices ${\bf H}_{ij}, \, \forall \,\, j\neq i$ and feeds back the corresponding information to the transmitters so that every transmitter is capable of solving the alignment problem.
In this section we consider perfect CSI feedback in order to highlight the intuition behind the dimensionality reduction associated with the proposed feedback scheme.  We will further assume that $(K-1)M \geq N$, which represents the cases of interest where interference would occupy all dimensions of the receive subspace in the absence of alignment.

For reference, let us first consider the case where the channel matrices  ${\bf H}_{ij}, \, \forall \,\, j\neq i$ themselves are known perfectly at the transmitter. The precoders ${{{\bf V}}_i}$, $i=1\ldots K$ must be designed to align the interference at each receiver into a $N-d$ dimensional space, in order to achieve $d$ interference-free dimensions per user.   
A solution to the IA problem exists (see \cite{Yetis_Gou_Jafar_Kayran_feasibility_of_IA09} and more recently \cite{Razaviyayn_etal_DoF_MIMO_IA_2011,Bresler_Cartwright_Tse_settling_feasibility_of_IA_symmetric_square_2011} for feasibility criteria  -- here we will assume that the dimensions and the considered channel realizations are such that the problem is feasible almost surely (a.s.)) iff there exist full rank precoding matrices ${ {\bf V}}_j ,\, j=1,\ldots,K$ and projection matrices ${\bf U}_i \in {\mathbb{C}^{N \times d }}, \,i=1,\ldots,K$ such that
\begin{eqnarray} 
&&{\bf U}_i^{\rm H}{\bf H}_{ij}{\bf V}_j={\bf 0} \,\,\, \,\  \,\,\,\,\, \forall i,j \in \{1,\ldots,K\}, \,\, j \neq i,  \quad \mathrm{and} \label{ia_definition1}
\\
&&{\rm rank}\left({\bf U}_i^{\rm H}{\bf H}_{ii}{\bf V}_i\right)=d. \label{ia_definition2} 
\end{eqnarray}

At this point, some remarks are in order. As pointed out in \cite{Gou_Jafar_DoF_MIMO_Kuser_IC_IT2010}, the difficulty in finding an IA solution typically lies in solving eq.~\eqref{ia_definition1}, while \eqref{ia_definition2} is fulfilled a.s. under the prevailing channel assumptions for any choice of full-column rank ${\bf U}_i$, ${\bf V}_j$ matrices.  We also remark that despite the symmetry of eq.~\eqref{ia_definition1} with respect to transposition, only the precoders are required to be known at the transmitters; for a given set of precoders ${\bf V}_1,\ldots, {\bf V}_K$, the mere knowledge of the \emph{existence} of full-column rank matrices ${\bf U}_1,\ldots, {\bf U}_K$ fulfilling \eqref{ia_definition1} is sufficient to conclude that the precoders are interference-aligning. These considerations lead us to introduce the following definition:
\begin{definition}[IA precoders]
The full-column rank precoders ${\bf V}_1,\ldots, {\bf V}_K$ are interference-aligning for the considered MIMO IC iff there exist full-column rank matrices ${\bf U}_1,\ldots, {\bf U}_K$ fulfilling \eqref{ia_definition1}.
\end{definition}

\subsection{Proposed Grassmannian feedback scheme}

In order to introduce our proposed scheme, let us note that \eqref{ia_definition1} can be rewritten from the point of view of receiver $i$ in the form
\begin{equation} \label{concatenated_IA_condition}
{\bf U}_i^{\rm H}{\bf H}_{i}{\bf V}_{-i}={\bf 0}  \,\,\, \,\,\, \forall i \in \{1,\ldots,K\},
\end{equation}
in which ${\bf V}_{-i}={\rm Bdiag}({\bf V}_1,\ldots,{\bf V}_{i-1},{\bf V}_{i+1},\ldots,{\bf V}_K) \in \mathbb{C}^{(K-1)M \times (K-1)d}$ is the block-diagonal concatenation of the precoders and ${\bf H}_i=[{\bf H}_{i\,1},\ldots,{\bf H}_{i\,i-1},{\bf H}_{i\,i+1},\ldots,{\bf H}_{i\,K}] \in \mathbb{C}^{N \times (K-1)M}$ is the concatenation of the channel matrices of all interfering links ending at receiver $i$, excluding the direct link. 
The proposed feedback scheme consists for each receiver $i$ in feeding back only the row space of ${\bf H}_i$. Our first result consists in stating that this information is sufficient to perform IA:
\begin{lemma}\label{lem1}
In order for the IA computation unit to compute interference-aligning precoders ${\bf V}_1,\ldots, {\bf V}_K$, it is sufficient that each receiver $i\in\{1,\ldots,K\}$ feeds back a point on the Grassmann manifold ${\mathcal G}_{(K-1)M,N}$ representing the row space of ${\bf H}_i$.
\end{lemma}
\begin{proof}
Let us consider perfect feedback of the row space of ${\bf H}_i, \, \forall i\in\{1,\ldots,K\}$. Practically, since a linear subspace can be represented by any matrix whose columns span the same space, the Grassmannian feedback considered here can be considered to take the form of the availability at the IA computation unit of a matrix ${\bf F}_i$ of dimensions $(K-1)M \times N$ whose columns span the same subspace as the columns of ${\bf H}_i^{\rm H}$ (we assume that ${\bf H}_i^{\rm H}$ has full column rank, which is a.s. the case for generic channel coefficients). We now show that the IA transmit precoders computed by assuming ${\bf F}_i^{\rm H}$ as channel coefficients are interference-aligning for the true channel as well.

Let us consider an IA solution based on ${\bf F}_i^{\rm H}$, i.e. assume that there exist full-rank matrices ${\bf U}_i$ and ${\bf V}_{i}$ such that the following equation (similar to \eqref{concatenated_IA_condition}),
\begin{equation}\label{Grassmann_feedback_IA_solution}
{\bf U}_i^{\rm H} {\bf F}_i^{\rm H}{\bf V}_{-i}={\bf 0},
\end{equation}
is fulfilled for all $i \in \{1,\ldots,K\}$. Note that since ${\bf H}_i$ and ${\bf F}_i^{\rm H}$ have the same dimensions, the feasibility (a.s.) of IA according to \eqref{concatenated_IA_condition} and \eqref{Grassmann_feedback_IA_solution} is identical. 
Furthermore, since the columns of ${\bf H}_i^{\rm H}$ and ${\bf F}_i$ span the same $N$-dimensional subspace, there exists an invertible $N \times N$ matrix ${\bf C}_i$ such that ${\bf H}_i^{\rm H} =  {\bf F}_i {\bf C}_i$. Clearly, 
\begin{eqnarray}\label{Grassmann_feedback_eq}
\eqref{Grassmann_feedback_IA_solution} &\Leftrightarrow &{\bf U}_i^{\rm H} {\bf C}_i^{\rm -H} {\bf C}_i^{\rm H} {\bf F}_i^{\rm H}{\bf V}_{-i}={\bf 0}  \\
& \Leftrightarrow &  \left( {\bf C}_i^{\rm -1} {\bf U}_i\right)^{\rm H} {\bf H}_i {\bf V}_{-i}={\bf 0}.\label{IA_transformed_filters}
\end{eqnarray}
Comparing to \eqref{ia_definition1}, eq.~\eqref{IA_transformed_filters} shows that the rank-$d$ matrices ${\bf C}_i^{\rm -1} {\bf U}_i$, $i \in \{1,\ldots,K\}$ cancel the interference at all receivers, i.e. the transmit precoders ${\bf V}_1,\ldots {\bf V}_K$ forming the block-diagonal of ${\bf V}_{-i}$ are interference-aligning over the true channels. 
\end{proof}

\subsection{Feedback dimension analysis}
\label{section_feedbackdimension}

As already noted, the CSI feedback scheme considered here, is analogous to feeding back a single point on the Grassmann manifold ${\mathcal G}_{(K-1)M,N}$ for each one of the $K$ users. 
Using the fact that the real dimension of ${\mathcal G}_{n,d}$ is $2d(n-d)$ for any $d\leq n$  \cite{dai}, the real dimension of the feedback variable in the strategy of Lemma~\ref{lem1} is $N_{\rm G}=2N((K-1)M-N)$. For comparison, let us consider the following alternative CSI representations:
\begin{itemize}
\item {\bf Full channel matrix (FCM):} for a given receiver $i$, the $K-1$ channel matrices ${\bf H}_{ij}$, $j \neq i$ appearing in \eqref{ia_definition1} taken together have real dimension $N_{\rm FCM}=2(K-1)MN$.
\item {\bf Individually normalized channel matrices (INM):} in \cite{KimMoonLeeLee_CSI_quantization_MIMOIC_TWC2012}, it is proposed to independently vectorize and normalize the matrices representing the channels from each interferers. At each receiver $i$, this technique yields $K-1$ unit-norm vectors ${\bf z}_{ij}=\frac{{\rm vec}({\bf H}_{ij})}{||{\rm vec}({\bf H}_{ij})||_2}$, $j\neq i$, which are subsequently quantized jointly on the composite Grassmann manifold ${\mathcal G}_{M N,1}^{K-1}$. The real dimension of this manifold is $N_{\rm INM}=2(K-1)(MN-1)$ \cite{rajac}.
\item {\bf Jointly normalized channel matrices (JNM)}\footnote{This approach was proposed by an anonymous reviewer of a previous version of this paper. We thank the reviewer for his suggestion.}: noting that \eqref{concatenated_IA_condition} can be rewritten as $\left({\bf V}_{-i}^{\rm H} \otimes {\bf U}_i^{\rm H}\right) {\rm vec}({\bf H}_{i})={\bf 0}$, this approach consists in quantizing ${\rm vec}({\bf H}_{i})/||{\rm vec}({\bf H}_{i})||_2$ on ${\mathcal G}_{(K-1)M N,1}$. The real dimension of the fed back variable for this case is $N_{\rm JNM}=2((K-1)M N-1)$.
\end{itemize}
It is straightforward to establish that $N_{\rm INM} \leq N_{\rm JNM} \leq  N_{\rm FCM}$ for all meaningful cases ($K\geq 2$). Furthermore, $N_{\rm G} < N_{\rm INM}$ iff $N^2 > K-1$. Note that this condition holds independently of the number of transmit antennas. In the particular case of a square system ($M=N$), we have the following result:
\begin{lemma} \label{lemma_NG_NINM}
In a square system, if IA is feasible, then $N_{\rm G} < N_{\rm INM}$, i.e. the proposed scheme always requires strictly less real dimensions than FCM, INM or JNM.
\end{lemma}
\begin{proof}
A necessary condition for IA to be feasible is \cite{Yetis_Gou_Jafar_Kayran_feasibility_of_IA09}
\begin{equation}\label{eqfs}
d\leq \frac{M+N}{K+1}.
\end{equation}
Together with the assumption that $M=N$ and using the fact that $d\geq1$, \eqref{eqfs} yields
\begin{equation}
K\leq \frac{2N}{d}-1< 2N. \label{eq_Km2N}
\end{equation}
Another necessary condition for IA feasibility is $N\geq 2d$, therefore $N>1$ and consequently $2N<N^2+1$. Combining with \eqref{eq_Km2N}, we obtain $K<N^2+1$,
which is equivalent to $N_{\rm G} < N_{\rm INM}$.
\end{proof}
\mbox{}\\

Note that the feedback scheme outlined here for the MIMO IC is in fact directly applicable to many other channel models where IA has been proposed, such as interfering multiple-access channels \cite{Suh_Tse_IA_cellular_allerton08,Sun_Liu_Zhu_DoF_cellular_network_Science_China_2010}, interfering broadcast channels \cite{Suh_Ho_Tse_downlink_IA_2010,Park_Lee_interfering_MISO_Globecom09}, as well as partially connected interference networks \cite{Guillaud_Gesbert_Globecom2011,Lee_Park_Kim_partially_connected_MIMO_IC_Globecom09}. \\

\section{Quantized CSI Feedback}
\label{sec:LF}
In this section we introduce a transmission scheme where the alignment equations are solved based on the (error-free) feedback of a quantized version of the CSI, based on the Grassmannian representation from Section~\ref{sec:IA}. For that scheme, we show in Section~\ref{section_proposed_quantized_feedback} how inter-user interference is related to the CSI codebook size, and characterize the scaling of the codebook size which ensures that the inter-user interference power remains bounded at high SNR.
For comparison, in Section~\ref{quantized_naive_feedback}, we provide a similar analysis for the INM technique. 

\subsection{Quantized feedback for the proposed scheme}
\label{section_proposed_quantized_feedback}
Let us assume that receiver $i$ knows perfectly the state of its channels from all interfering transmitters, i.e. the coefficients of ${\bf H}_i$, and performs the economy-size QR decomposition ${\bf H}_i^{\rm H}={\bf F}_i{\bf C}_i$, where ${\bf F}_i$ is a $(K-1)M \times N$ truncated unitary matrix, and ${\bf C}_i$ is $N\times N$ and a.s. invertible, under the prevailing channel assumptions.
The use of the QR decomposition is a particular case of the decomposition used in the proof of Lemma~\ref{lem1}: it ensures that ${\bf H}_i^{\rm H}$ and ${\bf F}_i$ have the same column space, and adds the requirement that the columns of ${\bf F}_i$ are orthonormal, which will simplify the subsequent analysis.
According to the proposed scheme, receiver $i$ quantizes the subspace spanned by the columns of ${\bf F}_i$ using $B_{\rm G}$ bits and feeds the index of the quantized codeword back to the unit in charge of computing the ${\bf V}_i$'s. We further assume that the receivers and the computation unit share a predefined codebook\footnote{For notational simplicity we omit the dependency of $\mathcal {S}$ on $i$, however the proposed analysis generalizes trivially to cases where ${\mathcal S}$ and $B_{\rm G}$ are different across the receivers, as will be seen in Section \ref{Asydof}.} ${\mathcal S}=\{{\bf S}_1,\ldots,{\bf S}_{2^{B_{\rm G}}}\}$ which is composed of $2^{B_{\rm G}}$ truncated unitary matrices of size $(K-1)M\times N$ and is designed via Grassmannian subspace packing \cite{Conway1996_Grassmannian_packing}. The quantized codeword at receiver $i$ is the point in ${\mathcal S}$ closest to ${\bf F}_i$, i.e.
\begin{equation}\label{eq_quantizer}
{\hat {\bf F}}_i=\mathrm{arg} \min_{{\bf S} \in {\mathcal S}} \,\,\,d_{c}({\bf S},{\bf F}_i)
\end{equation}
in which $d_c({\bf X},{\bf Y})=\frac{1}{\sqrt{2}}\left|\left| {\bf X}{\bf X}^{\rm H} - {\bf Y}{\bf Y}^{\rm H}\right|\right|_{\mathrm{F}}$ is the chordal distance between ${\bf X}$ and ${\bf Y} $ in ${\mathcal G}_{(K-1)M,N}$ \cite{rajathesis}. \\
\\
Let us consider the scheme where the interference alignment problem is solved at the IA computation unit based on the quantized CSI $\{{\hat {\bf F}}_i^{\rm H}\}_{i=1}^K$, yielding full-column rank matrices $(\{{\bf V}_i\}_{i=1}^K,\{{\tilde {\bf U}}_i\}_{i=1}^K)$ fulfilling
\begin{equation}\begin{split}\label{e8}
{\tilde {\bf U}}_i^{\rm H}{\hat {\bf F}}_i^{\rm H}{\bf V}_{-i}={\bf 0}, \,\,\,\forall i \in \{1,\ldots,K\}.   
\end{split}\end{equation}
At receiver $i$, 
inspired by the perfect feedback situation, we consider the receive filter ${\bf G}_i={\bf C}_i^{-1}{\bf F}_i^{\rm H}{\hat {\bf F}}_i{\tilde {\bf U}}_i$ \footnote{We note that if the quantization error is null, i.e. $d_c({\hat {\bf F}}_i,{\bf F}_i)=0$, then ${\bf F}_i^{\rm H}{\hat {\bf F}}_i$ is a unitary matrix corresponding to the uncertainty between the CSI encoder (at the receiver) and decoder (at the IA computation unit) in the matrix representation of the subspace being fed back.}.
Let ${{\bf{y}}'_i}$ denote the received signal at receiver $i$ after processing by ${\bf G}_i$:
\begin{equation}\label{Ehh1}
	{{\bf{y}}'_i}={\bf G}_i^{\rm H}{{\bf{y}}_i} = {\bf G}_i^{\rm H}{{{{\bf H}}}_{ii}}{{{\bf V}}_i}{{\bf{x}}_i} + {\bf e}_i  + {\bf G}_i^{\rm H}{{\bf{n}}_i},
\end{equation}
where the term
\begin{equation}\begin{split}
{\bf e}_i= \sum_{ \substack{1\leq j\leq K\\ j \ne i}}{{\bf G}_i^{\rm H}{{{{\bf H}}}_{ij}}{{{{\bf V}}}_j}{{\bf{x}}_j}}
\end{split}\end{equation}
is the interference leakage due to the imperfect CSI.\\

Generally speaking, the aim of our analysis is to provide sufficient conditions on the CSI quantization accuracy to ensure that ${\mathcal I }({{\bf{x}}_i};{{\bf{y}}_i})$ grows with $d \log P$ (see Section \ref{sec:rateanal}); ${\bf G}_i$ and ${{\bf{y}}'_i}$ are merely intermediate variables used to establish information-theoretic inequalities. In a practical system, we expect the equalizer ${\bf G}_i$ to be computed through classical channel estimation and equalization techniques -- we omit these details here.\\

In the remainder of this section, we will focus on establishing bounds on the interference power $L_i= {\rm tr}\left( {\rm E}_{\bf x}({\bf e}_i{\bf e}_i^{\rm H}) \right)$; these results will be instrumental in proving our DoF result in Section~\ref{sec:rateanal}. We first establish in Lemma~\ref{thm_boundedleakage} and Corollary~\ref{corl1} the growth rate of the number of feedback bits with the SNR which guarantees that $L_i$ remains bounded by a constant regardless of $P$ when $P\rightarrow \infty$. 

\begin{lemma}\label{thm_boundedleakage}
The interference leakage power (due to imperfect CSI) at receiver $i$ can be bounded as
\begin{equation}\begin{split}\label{e293}
L_i\leq \frac{8P}{{(c\, 2^{B_{\rm G}})}^{\frac{2}{N_{\rm G}}}} \left(1+ o\left(2^{-\frac{B_{\rm G}}{N_{\rm G}}}\right)\right)
\end{split}\end{equation}
where $N_{\rm G}=2N((K-1)M-N)$ is the real dimension of ${\mathcal G}_{(K-1)M,N}$ introduced before, and $c$ is the coefficient of the ball volume in the Grassmann manifold,
\begin{equation}\label{cequ}
c\triangleq \frac{1}{\big(N((K-1)M-N)\big)!}\frac{\prod_{i=1}^N\big((K-1)M-i\big)!}{\prod_{i=1}^N\big(N-i\big)!}.
\end{equation} 
\end{lemma}
\begin{proof}
See appendix \ref{appendix_leakage_bound}.
\end{proof}

\begin{corollary}\label{corl1}
Quantizing CSI with 
\begin{equation}\label{eqa}
B_{\rm G} = N((K-1)M-N)\log P
\end{equation}
bits is sufficient to keep the interference leakage $L_i$ bounded by a constant for arbitrarily large $P$.
\end{corollary}
\begin{proof}
From \eqref{e293}, since $o\left(2^{-\frac{B_{\rm G}}{N_{\rm G}}}\right) \rightarrow 0$ for large $P$, it is obvious that $L_i$ is bounded by a constant if $2^{\frac{2B_{\rm G}}{N_{\rm G}}}$ scales at least linearly with $P$; in particular this holds for
\begin{equation}
B_{\rm G}=\frac{N_{\rm G}}{2}\log  P=N((K-1)M-N)\log P.
\end{equation}
\end{proof}

\subsection{Quantized feedback for the INM method}
\label{quantized_naive_feedback}
For comparison, let us now consider quantization for the INM method\footnote{The authors of \cite{KimMoonLeeLee_CSI_quantization_MIMOIC_TWC2012} attribute this method to \cite{rajac}. Although quantization bounds for the composite Grassmann manifold are presented in \cite{rajac}, we note that the (frequency-selective) channel model in that paper is different from the flat-fading model considered here and in \cite{KimMoonLeeLee_CSI_quantization_MIMOIC_TWC2012}, and therefore the results are not immediately comparable.} chosen as a baseline in  \cite{KimMoonLeeLee_CSI_quantization_MIMOIC_TWC2012}. We recall that in that case, at receiver $i$ the matrices representing the channels from the interferers are vectorized and normalized independently, yielding $K-1$ unit-norm vectors ${\bf z}_{ij}=\frac{{\rm vec}({\bf H}_{ij})}{||{\rm vec}({\bf H}_{ij})||_2}$, $j\neq i$. ${\bf Z}_{i}=[{\bf z}_{i1},\ldots,{\bf z}_{i\,i-1},{\bf z}_{i\,i+1},\ldots,{\bf z}_{iK}] \in {\mathcal G}_{M N,1}^{K-1}$ is subsequently quantized according to
\begin{equation}
{\bf \hat Z}_{i}=\mathrm{arg} \min_{{\bf T} \in {\mathcal T}} \,\,\,D_{c}({\bf T},{\bf Z}_{i}),
\end{equation}
where $D_c({\bf T},{\bf Z}_{i})=\sqrt{\rm{tr}\left({\bf I}_{K-1}-{\bf T}^{\rm H} {\bf Z}_{i} \right)}$ is the chordal distance defined for the composite Grassmann manifold. Let $B_{\rm INM}$ denote the number of feedback bits, i.e. $|\mathcal T|=2^{B_{\rm INM}}$.
At the transmitter side, the columns of ${\bf \hat Z}_{i}=[{\bf \hat z}_{i1},\ldots,{\bf \hat z}_{i\,i-1},{\bf \hat z}_{i\,i+1},\ldots,{\bf \hat z}_{iK}]$ are used to reconstruct the quantized CSI: the channel matrices ${ \bf \hat H}_{ij}$ used for the computation of the precoders are such that
${\rm vec}({ \bf \hat H}_{ij}) ={\bf \hat z}_{ij}, \, \forall i\neq j$.
The interference alignment problem is then solved based on ${ \bf \hat H}_{ij}$ to find $(\{{\bf V}_i\}_{i=1}^K,\{{\tilde {\bf U}}_i\}_{i=1}^K)$ fulfilling
\begin{equation}\begin{split}\label{err8}
{\tilde {\bf U}}_i^{\rm H}{ \bf \hat H}_{ij}{\bf V}_{j}={\bf 0}, \,\,\,\forall i,j \in \{1,\ldots,K\},\,\, j\neq i.   
\end{split}\end{equation}

We now show that the leakage ${\bar L}_i=\frac{P}{d}||{\tilde {\bf U}}_i^{\rm H}{ \bf H}_{ij}{\bf V}_{j}||_{\rm F}^2 $ (using the true channel matrices) can remain bounded for arbitrarily large transmit power $P$ under certain conditions. This is the object of Lemma~\ref{inlem2}, where we establish the scaling of $B_{\rm INM}$ with $P$ required to achieve bounded interference leakage under this scheme.

\begin{lemma}\label{inlem2} 
Using the INM quantization scheme, quantizing ${\bf Z}_{i}$ with $B_{\rm INM}=\frac{1}{2}N_{\rm INM}{\rm log} P=(K-1)(MN-1){\rm log} P$ bits is sufficient to keep ${\bar L}_i$ bounded for arbitrarily large $P$. 
\end{lemma}
\begin{proof}
See appendix \ref{appendix_scaling_naive}.
\end{proof}

Comparing the above result with the scaling obtained in Corollary~\ref{corl1} for the proposed scheme indicates that at high SNR, $B_{\rm G} < B_{\rm INM}$ (i.e. the proposed method strictly outperforms INM) iff $N_{\rm G} < N_{\rm INM}$. As already analyzed in Lemma~\ref{lemma_NG_NINM}, this condition is fulfilled for many case of practical interest.\\

\section{Achievable DoF and Rate Analysis}
\label{sec:rateanal}

\subsection{Rate and DoF loss due to CSI Quantization}

In the previous section, we have used interference leakage as a proxy to evaluate how the quality of the available CSI influences alignment. Note however that having a bounded interference leakage is not sufficient in itself to ensure that the full DoF is achieved for asymptotically large $P$ -- in fact, the power of the signal of interest remaining after processing by the receive filter (eq.~\eqref{Ehh1}) could remain bounded too, or the equivalent channel ${\bf G}_i^{\rm H}{{{{\bf H}}}_{ii}}{{{\bf V}}_i}$ could be rank-deficient.  We now show that this is almost surely not the case, and that the proposed CSI quantization scheme achieves the same DoF as IA under the perfect CSI assumption, provided that the proper scaling of $B_{\rm G}$ with $P$ is respected:

\begin{theorem}\label{thm_scalingd}
If IA with $d$ DoF is feasible, the proposed CSI quantization scheme achieves $d$ DoF for almost all channel realizations if $B_{\rm G}$ is scaled according to \eqref{eqa}.
\end{theorem}

\emph{Remark 1:} Theorem~\ref{thm_scalingd} is not restricted to a particular distribution of the channel coefficients. The restriction to ``almost all'' channel realizations is due to the fact that under the assumptions of Section~\ref{sec:model}, there can exist a vanishing set of channel realizations for which \eqref{ia_definition2} is not fulfilled;  this is also the case when perfect CSI is considered \cite{Gou_Jafar_DoF_MIMO_Kuser_IC_IT2010}, and is unrelated to the proposed quantization scheme.\\

\emph{Remark 2:} The transmission scheme considered here is based on truncated unitary precoders ${\bf V}_j$, and therefore the transmitted signal is spatially white inside the $d$-dimensional subspace defined by the precoder. Clearly, this is suboptimal for finite values of the SNR, and spatial waterfilling in addition to IA would bring in performance improvement for $d>1$. However, we remark that the performance gains of waterfilling vanish at asymptotically high SNR, provided that the channel is not rank deficient \cite{Moon_waterfilling_high_SNR_TCom_2011}. Therefore, the asymptotic analysis of this section holds regardless of whether spatial waterfilling is used in addition to IA or not.\\

Theorem \ref{thm_scalingd} states that $\lim_{P\rightarrow \infty} \frac{R_i}{\log P}=d$ a.s.; in order to show this, we require a few intermediate results. Let us define the following values: $R_i\triangleq {\mathcal I }({{\bf{x}}_i};{{\bf{y}}_i})$, $R_i'\triangleq {\mathcal I }({{\bf{x}}_i};{{\bf{y}}_i'})$ and $R_i''\triangleq \log \left| {{\bf G}_i^{\rm H}{\bf G}_i + \frac{P}{d}{\bf Q}_{\rm S}^i} \right|$ where ${\bf Q}_{\rm S}^i={\bf G}_i^{\rm H}{\bf H}_{ii}{\bf V}_{i}{\bf V}_{i}^{\rm H}{\bf H}_{ii}^{\rm H}{\bf G}_i$ is the covariance of the signal of interest.
From the data processing inequality and the definition of ${\bf{y}}_i'$, we have immediately that $R_i \geq R_i'$. In what follows, we will successively show that $R_i'' - R_i'$ remains bounded from above if $B_{\rm G}$ is scaled according to \eqref{eqa} (Lemma~\ref{lemma_LB_rate}), and that $\lim_{P\rightarrow \infty} \frac{R_i''}{\log P}=d$ (Lemma~\ref{lem_dof_no_intf}).
Let us start with the first result. 
Since all signal and noise terms are Gaussian circularly symmetric, we have
\begin{equation}
R_i'=\log \left| {{\bf G}_i^{\rm H}{\bf G}_i + \frac{P}{d}({\bf Q}_{\rm S}^i+{\bf Q}_{\rm I}^i)} \right|-\log \left| {{\bf G}_i^{\rm H}{\bf G}_i + \frac{P}{d}{\bf Q}_{\rm I}^i} \right| \label{rate_quantized}
\end{equation}
in which ${\bf Q}_{\rm I}^i={\bf G}_i^{\rm H}{\bf H}_i{\bf V}_{-i}{\bf V}_{-i}^{\rm H}{\bf H}_i^{\rm H}{\bf G}_i$ is the covariance of the residual interference.

\begin{lemma} \label{lemma_LB_rate}
Under the quantization scheme of Section~\ref{section_proposed_quantized_feedback}, the difference between $R_i'$ and $R_i''$ can be bounded as
\begin{equation}\begin{split}\label{in24}
R_i''-R_i' \leq d \log \left(||{\bf C}_i^{-1}||_2^2+ \frac{8P}{d {(c\, 2^{B_{\rm G}})}^{\frac{2}{N_{\rm G}}}}\left(1+o\left(2^{-\frac{B_{\rm G}}{N_{\rm G}}}\right)\right) \right).
\end{split}\end{equation}
\end{lemma}
\begin{proof}
See Appendix~\ref{appendix_proof_LB_rate}.
\end{proof}

\begin{lemma}\label{lem_dof_no_intf}
Under the proposed CSI quantization scheme, there exists a series of codebooks of increasing size following \eqref{eqa} for $P\rightarrow \infty$ s.t. $\lim_{P\rightarrow \infty} \frac{R_i''}{\log P}=d$ a.s.
\end{lemma}
\begin{proof}
See Appendix~\ref{appendix_proof_dof_no_intf}.
\end{proof}
\mbox{}\\*

We are now in the position to prove Theorem \ref{thm_scalingd}:
\begin{proof}
Substituting $P=2^{\frac{2B_{\rm G}}{N_{\rm G}}}$ in the result of Lemma~\ref{lemma_LB_rate} yields
\begin{align}
R_i'& \geq R_i''-d \log \left(||{\bf C}_i^{-1}||_2^2+  \frac{8}{c^{2/N_{\rm G}} d}\left(1+o\left(2^{-\frac{B_{\rm G}}{N_{\rm G}}}\right)\right) \right).
\end{align}
As $P\rightarrow \infty$ and with $B_{\rm G}$ following \eqref{eqa}, the argument of the logarithm remains bounded by a constant, therefore
\begin{equation}
\lim_{P \rightarrow \infty} \frac{R_i'}{\log P} \geq  \lim_{P \rightarrow \infty} \frac{R_i''}{\log P}   = d \quad \mathrm{a.s.} \label{rank_ineq}
\end{equation}
using the result from Lemma~\ref{lem_dof_no_intf}.
\end{proof}
\mbox{}\\*

\subsection{Per-User DoF for Asymmetric Feedback}\label{Asydof}
An interesting consequence of the rate-loss analysis conducted previously can be observed when each receiver uses its own scaling of the CSI quantization codebook size with $P$. Formally, let $B_{\rm G}^i$ denote the number of bits used by receiver $i$ to quantize ${\bf F}_i$.
\begin{corollary}
If $B_{\rm G}^i$ scales with $P$ such that
\begin{equation}
\alpha_i \triangleq \lim_{P \rightarrow \infty} \frac{B_{\rm G}^i}{N_{\rm G}/2 \cdot \log P}
\end{equation}
exists and is finite, then the DoF achievable by user $i$ is 
\begin{equation}
d_i^{\rm q}\geq d_i^{\rm p} \min\left( \alpha_i,1\right),
\end{equation}  
where $d_i^{\rm p}$ is the achievable DoF of this user with perfect CSI.
\end{corollary}
\begin{proof}
The proof follows simply from \eqref{in24} by taking the limit of the lower bound when $P \rightarrow \infty$.
\end{proof}

Practically, this means that the DoF achieved by a given user is independent of the quality of the feedback provided by the other users, and depends only on the scaling of its own feedback.
This observation, obtained here for IA precoding, is consistent with the results obtained in \cite{deKerret_Gesbert_DoF_distributed_CSI_IT12} for centralized schemes using different precoding schemes such as zero-forcing.\\

\subsection{Average Rate Loss under Random Vector Quantization}
Note that the results established so far hold for any codebook obtained by sphere-packing. Let us now briefly depart from this assumption, and consider RVQ instead. In that case, the previous results do not apply: the random choice of the codebook can lead to arbitrarily bad performance regardless of $B_{\rm G}$, and bounding the performance loss uniformly over all  codebooks is impossible. A more relevant performance metric for RVQ is the average sum rate over all possible codebooks. We have the following result:
\begin{theorem}\label{thm_rvq}
Provided that the codebook $\mathcal{S}$ is generated from independent realizations of a random process uniformly distributed over ${\mathcal G}_{(K-1)M,N}$, the expectation over $\mathcal{S}$ of $R_i$ is lower bounded as
\begin{equation}\begin{split}\label{in28}
{\rm E}_{{\mathcal S}}(R_i)
& \geq R_i''-d \log \left (||{\bf C}_i^{-1}||_2^2+\frac{2P}{d} \frac{\Gamma ({\frac{2}{N_{\rm G}}})}{{\frac{N_{\rm G}}{2}}{(c\, 2^{B_{\rm G}})}^{\frac{2}{N_{\rm G}}}}\right ),
\end{split}\end{equation}
where $\Gamma(\cdot)$ denotes the Gamma function. 
\end{theorem}
\begin{proof}
See Appendix \ref{appendix_proof_rvq}.
\end{proof} 

\mbox{}\\*

\section{Simulation Results}
\label{sec:simulation}

This section presents simulations that numerically validate the results hitherto established. 
Note that constructing good Grassmannian packings for arbitrary dimensions is difficult \cite{Dhillon_Grassmannian_packings_via_projection_expmat08}; therefore, in our simulations for relatively small codebook sizes (up to $2^{15}$) we resort to random codebooks in place of sphere-packing codebooks. Note that the performance expected from RVQ codebooks constitutes a lower bound to the performance of sphere-packing codebooks; however as we shall see, in our simulations, RVQ codebooks attain the performance predicted for the sphere-packing codebooks.  These results are presented in Section~\ref{section_simul_rvq}.

For larger codebooks ($B_{\rm G}>15$), even RVQ is not tractable due to the complexity of the exhaustive search through $\mathcal{S}$ in \eqref{eq_quantizer}.
Due to the lack of structured codebooks allowing a tractable implementation of the quantizer, the performance obtained for larger codebooks is extrapolated by using a perturbation method based on the analytical characterization of the distribution of the quantization error, the details of which being presented in Section~\ref{GMperturbation}.

\subsection{Performance results using RVQ}
\label{section_simul_rvq}
In this section, we evaluate the performance of the quantization scheme of Section~\ref{section_proposed_quantized_feedback} with RVQ codebooks. The performance metric is the sum rate evaluated through Monte-Carlo simulations. 
The sum rate achievable over the MIMO IC using interference alignment precoders under the assumption that the input signals are Gaussian can be written as
\begin{equation}\label{E11}
	{R_{\rm sum}} =    \sum\limits_{i= 1}^K \log \left| {{{\bf I}_{{N}}} + \frac{P}{d}\sum\limits_{j = 1}^K {{{\bf H}_{ij}}{{\bf {\bf V}}_j}{{\bf {\bf V}}_j^ {\rm H}}{\bf H}_{ij}^{\rm H}} } \right|  -   \sum_{i=1}^K \log \left| {{{\bf I}_{{N}}} + \frac{P}{d}\sum\limits_{j=1, j \ne i}^K {{{\bf H}_{ij}}{{\bf {\bf V}}_j}{{\bf {\bf V}}_j^ {\rm H}}{\bf H}_{ij}^{\rm H}} } \right|.
\end{equation}
A $3$-user IC with $M=N=2$ antennas per node and $d=1$ data stream for each transmitter is considered. Entries of the channel matrices are generated according to $\mathcal{CN}(0,1)$ and the performance results are averaged over the channel realizations.
The method proposed in Section~\ref{section_proposed_quantized_feedback} is compared to the INM quantization method from Section~\ref{section_feedbackdimension}.

For the proposed method, the codebook entries are independent $(K-1)M \times N$ random truncated unitary matrices generated from the Haar distribution. For the INM method, random unit norm vectors are used in the codebook construction. Figure \ref{fig:sinr} shows the achievable sum rate versus transmit SNR for $B_{\rm G}=5$ and 10 feedback bits when the precoders are designed based on the quantized feedback.
Clearly the proposed scheme outperforms INM quantization for the same number of feedback bits. It can be also seen that for a fixed number of feedback bits, the sum-rate saturates at high SNR, while it grows unbounded (with the slope equal to the DoF) for the perfect CSI case.\\ 

\begin{figure}
  \centering
  \begin{tikzpicture}[scale=1]
    \renewcommand{\axisdefaulttryminticks}{4}
    \pgfplotsset{every major grid/.append style={densely dashed},every mark/.append style={solid}}
    \tikzstyle{every axis y label}+=[yshift=-10pt]
    \tikzstyle{every axis x label}+=[yshift=5pt]
    \pgfplotsset{every axis legend/.append style={cells={anchor=west},fill=white, at={(1,0)}, anchor=south east,{font=\scriptsize}}}
    \begin{axis}[ 
      xmin=0,
      ymin=2,
      xmax=40,
      ymax=13.5,
      bar width=3pt,
      grid=major,
      scaled ticks=true,
      ylabel={Average sum-rate [bits/s/Hz]},
      xlabel={{\rm SNR} [dB]}
      ]
 
     
     \addplot[red,dashed] plot coordinates{(0,2.7431)    (5,4.1824)    (10,5.4552)    (15,6.3308)    (20,6.8235)    (25,7.0655)    (30,7.1743)  (35,7.2202)    (40,7.2385)};
      
     \addplot[red,dashed,mark=o,every mark/.append style={solid}] plot coordinates{(0,2.8553)    (5,4.4862)    (10,5.9964)    (15,7.0588)    (20,7.6485)    (25,7.9210)    (30,8.0318)    (35,8.0735)    (40,8.0883)};

    \addplot[blue,smooth,mark=x] plot coordinates{(0,2.8696)    (5,4.4656)    (10,5.9372)    (15,6.9886)    (20,7.5938)    (25,7.8901)    (30,8.0197)   (35,8.0726)    (40,8.0938)};
         
     \addplot[blue,smooth,mark=diamond] plot coordinates{(0,3.0309)    (5,4.9674)    (10,6.9936)    (15,8.6285)    (20,9.6730)   (25,10.2307)   (30,10.4918)   (35,10.6027)   (40,10.6465)};

       \addplot[black,smooth,mark=triangle] plot coordinates{(0,3.2634)    (5,5.7347)    (10,9.1077)   (15,13.2112)          } ;

 \legend{ {$B_{\rm G}=5$, INM},{$B_{\rm G}=10$, INM},{$B_{\rm G}=5$, Proposed},{$B_{\rm G}=10$, Proposed}, {Perfect CSI}}
    \end{axis}
  \end{tikzpicture}
  \caption{Average $R_{\rm sum}$ for various quantization methods, for the $3$-user MIMO IC, $N=M=2$.}
  \label{fig:sinr}
\end{figure}
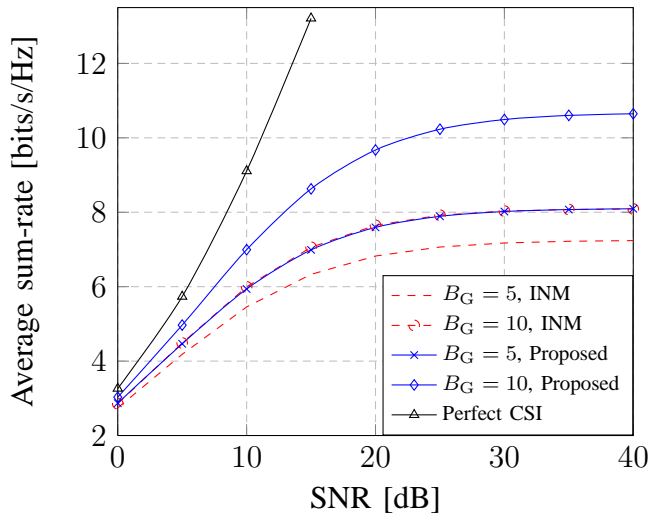

The sum rate in \eqref{E11} is achievable when optimum receivers (not including the projection filters ${\bf G}_i^{\rm H}$) are used at the receivers.
Since the achievable scheme in Section~\ref{sec:LF} is using the projection filters ${\bf G}_i^{\rm H}$, we evaluated the performance achieved by this scheme, defined as
\begin{equation}\label{E121}
	{R'_{\rm sum}} =\sum\limits_{i= 1}^K \log \left| {{\bf G}_i^{\rm H}{\bf G}_i + \frac{P}{d}\sum\limits_{j = 1}^K {{\bf G}_i^{\rm H}{{\bf H}_{ij}}{{\bf {\bf V}}_j}{{\bf {\bf V}}_j^ {\rm H}}{\bf H}_{ij}^{\rm H}{\bf G}_i} } \right| - \sum_{i=1}^K \log \left| {{\bf G}_i^{\rm H}{\bf G}_i + \frac{P}{d}\sum\limits_{j=1, j \ne i}^K {{\bf G}_i^{\rm H}{{\bf H}_{ij}}{{\bf {\bf V}}_j}{{\bf {\bf V}}_j^ {\rm H}}{\bf H}_{ij}^{\rm H}{\bf G}_i} } \right|.
\end{equation}
Results are provided in Figure \ref{fig:IA}. The slope of the curves at high SNR gives an indication of the achieved DoF.  
It is clear from Figure \ref{fig:IA} that  the slope of the sum-rate curve with quantized feedback matches that of perfect CSI when the number of feedback bits is scaled according to \eqref{eqa} (here we have used $B_{\rm G}=[0, \, 7, \, 13, \, 20, \, 26]$ bits and the corresponding powers $P=2^\frac{2B_{\rm G}}{N_{\rm G}}$). Conversely, when the codebook size is fixed, the performance always saturates at high SNR, with the achieved performance depending on the codebook size. Simulations were performed only up to $20 \,{\rm dB}$ SNR due to the complexity associated to the growth of the codebook size with $P$.

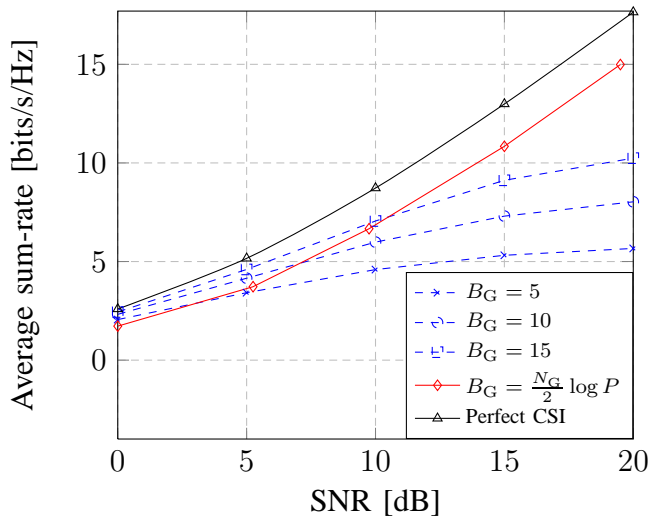
\begin{figure}
  \centering
  \begin{tikzpicture}[scale=1]
    \renewcommand{\axisdefaulttryminticks}{4}
    \pgfplotsset{every major grid/.append style={densely dashed},every mark/.append style={solid}}
    \tikzstyle{every axis y label}+=[yshift=-10pt]
    \tikzstyle{every axis x label}+=[yshift=5pt]
    \pgfplotsset{every axis legend/.append style={cells={anchor=west},fill=white, at={(1,0)}, anchor=south east,{font=\scriptsize}}}
    \begin{axis}[ 
      xmin=0,
      ymin=-4,
      xmax=20,
      ymax=17.7,
      bar width=3pt,
      grid=major,
      scaled ticks=true,
      ylabel={Average sum-rate [bits/s/Hz]},
      xlabel={{\rm SNR} [dB]}
      ]            
   
                 

     \addplot[blue,dashed,mark=x,every mark/.append style={solid}] plot coordinates{(0,2.0569)    (5,3.4207)    (10,4.5922)    (15,5.3169)    (20,5.6651)    (25,5.8046)    (30,5.8542)    (35,5.8707)    (40,5.8760)};
     
     \addplot[blue,dashed,mark=o,every mark/.append style={solid}] plot coordinates{(0,2.3220)    (5,4.1550)    (10,5.9851)    (15,7.3029)    (20,8.0284)    (25,8.3532)    (30,8.4785)    (35,8.5226)    (40,8.5373)};
    
     \addplot[blue,dashed,mark=square,every mark/.append style={solid}] plot coordinates{(0,2.4420)    (5,4.5997)    (10,7.0555)    (15,9.1196)    (20,10.2427)    (25,8.3532)    (30,8.4785)    (35,8.5226)    (40,8.5373)};
     
      \addplot[red,mark=diamond] plot coordinates{(0,1.7240)    (5.25,3.7240)    (9.75,6.6640)    (15,10.8420)    (19.5,14.9883)};
      
       \addplot[black,smooth,mark=triangle] plot coordinates{(0,2.5948)    (5,5.1614)    (10,8.7260)   (15,12.9892)  (20,17.6472)          } ;

 \legend{  {$B_{\rm G}=5$},{$B_{\rm G}=10$}, {$B_{\rm G}=15$}, {$B_{\rm G}=\frac{N_{\rm G}}{2}\log  P$},{Perfect CSI}}
    \end{axis}
  \end{tikzpicture}
  \caption{Sum-rate according to \eqref{E121} of the proposed method for different number of bits, for the $3$-user MIMO IC, $N=M=2$.}
  \label{fig:IA}
\end{figure}

\subsection{Perturbations on the Grassmann manifold}
\label{GMperturbation}
In order to validate the DoF results of Section~\ref{sec:rateanal}, an evaluation of the achieved sum-rate at high SNR is required. In order to deal with exponentially large codebooks, we propose to replace the quantization process with a perturbation which approximates the quantization error. In other words, we propose to replace ${\hat {\bf F}}_i$ by a matrix that can be computed directly by an appropriate perturbation of ${\bf F}_i$. This approach provides a good approximation of the achievable performance, while sparing the complexity associated with the codebook generation and the quantization in RVQ.\\

Let us consider a point on ${\mathcal G}_{n,p}$, represented by a $n \times p$ truncated unitary matrix $\bf F$.  Here, we assume that $n\geq2p$ (otherwise it is more efficient to consider the complementary $n-p$ dimensional subspace).
Since the columns of $\bf F$ are orthonormal, they can be completed to form an orthonormal basis of the $n$-dimensional space.
In fact, according to \cite{Barg}, any other point on ${\mathcal G}_{n,p}$ can be represented in the basis constituted by the columns of the unitary matrix ${\bf W} = \left[{\bf F} \,\, {\bf F}^{\rm c}\right]$ as
\begin{equation}\label{eq_Fbar}
{\bf \bar F}={\bf W} \begin{bmatrix} {\bf C} \\ {\bf S} \\ {\bf 0}_{n-2p} \end{bmatrix}\!\!,
\end{equation}
for some ${\bf F}^{\rm c}$ in the null space of $\bf F$ and  
\begin{equation}\label{i33}
{\bf C}=\left [ \begin{array}{ccc} \cos\theta_1 & \cdots & 0 \\ \vdots & \ddots & \vdots \\ 0 & \cdots & \cos\theta_p \end{array}\right ] \!\!, \,\,{\bf S}=\left [ \begin{array}{ccc} \sin\theta_1 & \cdots & 0 \\ \vdots & \ddots & \vdots \\ 0 & \cdots & \sin\theta_p \end{array}\right ] 
\end{equation}
where $\theta_1,\ldots,\theta_p$ are real angles. 
Clearly, for $\theta_1=\ldots=\theta_p=0$, we obtain ${\bf \bar F}=\bf F$. More generally, the squared chordal distance between the two points on ${\mathcal G}_{n,p}$ represented by $\bf F$ and ${\bf \bar F}$ is
\begin{equation}
r=d_c^2({\bf \bar F},{\bf F})=\sum_{i=1}^p \sin^2\theta_i \,.
\end{equation}

Therefore, in order to generate random perturbations of a certain chordal distance $\sqrt{r}$ from $\bf F$, we propose to generate random values for the angles $\theta_1,\ldots,\theta_p$ such that $\sum_{i=1}^p \sin^2\theta_i=r$, and to pick a random orthonormal basis ${\bf F}^{\rm c}$ of the null subspace of ${\bf F}$. The perturbed matrix is then computed using \eqref{eq_Fbar}. 

The histogram (not shown) of the squared quantization error $d_c^2({\bf \hat F},{\bf F})$ obtained from an implementation of the RVQ quantizer suggests that the Gaussian distribution is a good approximation for the probability density function of $r$. The parameters of this distribution can be obtained from \cite[Theorem~6]{rajathesis} which provides bounds on the $k$-th moment of the chordal distance $D^{(k)}={\rm E}_{{\mathcal S},{\bf F}}(d_c^k({\hat {\bf F}},{\bf F}))$.
Since those bounds are asymptotically tight when the codebook size increases, we arbitrarily choose to use the upper bound\footnote{Experiments have shown no noticeable performance difference when using the lower bound instead.} as an approximation of $D^{(k)}$, i.e. 
\begin{equation}
{\bar r} \triangleq  \frac{\Gamma ({\frac{2}{N_{\rm G}}})}{{\frac{N_{\rm G}}{2}}{(c\, 2^{B_{\rm G}})}^{\frac{2}{N_{\rm G}}}} \approx D^{(2)} 
\end{equation}
is the mean and 
\begin{equation}
\sigma^2_r \triangleq \frac{\Gamma ({\frac{4}{N_{\rm G}}})}{{\frac{N_{\rm G}}{4}}{(c\, 2^{B_{\rm G}})}^{\frac{4}{N_{\rm G}}}} - {\bar r}^2 \approx  D^{(4)}-(D^{(2)})^2 
\end{equation}
is the variance. We propose generate the values for $r$ according to $\mathcal{N}({\bar r},\sigma^2_r)$ truncated to $\mathbb{R}^+$.  This process is summarized in Algorithm~\ref{algo_perturbation}.

\begin{algorithm}[h!] 
  \caption{Generating random perturbations around ${\bf F}$} \label{algo_perturbation}
     \begin{itemize}
	     \item Draw a random realization of the squared chordal distance $r$ from $\mathcal{N}(\bar r,\sigma^2_r)$ 
	     \item If $r<0$, generate a new sample
	     \item Draw independent $s_1,\ldots , s_p$ uniformly from the interval $(0,1)$
	     \item Compute the angles $\theta_i=\sin^{-1}\left(\frac{s_i\sqrt{r}}{\sqrt{\sum_{i=1}^p s_i^2}}  \right)$
	     \item Generate a random orthonormal basis ${\bf F}^{\rm c}$ of the null space of ${\bf F}$
             \item Compute ${\bf \bar F}$ according to \eqref{eq_Fbar}.
     \end{itemize}
  \end{algorithm}

Simulations were performed in order to validate experimentally the perturbation method proposed above.
The sum-rate performance achieved by IA for the CSI obtained from the perturbation method is plotted against the performance obtained for the actual quantization scheme in Figure \ref{fig:perturb}. It is clear that the proposed perturbation method accurately approximates the Grassmannian quantization process, even for small codebooks. 

\subsection{Validation of the DoF results}
We now use the perturbation technique introduced in the previous section to analyze the CSI feedback scheme from Section~\ref{section_proposed_quantized_feedback} in the high SNR regime.
Figure \ref{fig:lb} depicts the sum rate performance using the perturbation method compared to perfect CSI and to the lower bound derived in \eqref{in28}. The slope of the sum rate at high SNR regime obtained for the quantizer with $B_{\rm G}=\frac{N_{\rm G}}{2}\log P$ bits is identical to that of perfect CSI, as is the case for the lower bound derived in \eqref{in28}. 
  \\

\begin{figure}
  \centering
  \begin{tikzpicture}[scale=1]
    \renewcommand{\axisdefaulttryminticks}{4}
    \pgfplotsset{every major grid/.append style={densely dashed},every mark/.append style={solid}}
    \tikzstyle{every axis y label}+=[yshift=-10pt]
    \tikzstyle{every axis x label}+=[yshift=5pt]
    \pgfplotsset{every axis legend/.append style={cells={anchor=west},fill=white, at={(1,0)}, anchor=south east,{font=\scriptsize}}}
    \begin{axis}[ 
      xmin=0,
      ymin=-1,
      xmax=40,
      ymax=15,
      bar width=3pt,
      grid=major,
      scaled ticks=true,
      ylabel={Sum-rate [bits/s/Hz]},
      xlabel={{\rm SNR} [dB]}
      ]            
      
                 

      
        \addplot[blue,mark=x,every mark/.append style={solid}] plot coordinates{(0,2.0834)    (5,3.4483)    (10,4.5916)    (15,5.2723)    (20,5.5859)    (25,5.7076)    (30,5.7502)    (35,5.7642)    (40,5.7687)};
     
     \addplot[blue,mark=o,every mark/.append style={solid}] plot coordinates{(0,2.3247)    (5,4.1639)    (10,5.9900)    (15,7.2750)    (20,7.9553)    (25,8.2481)    (30,8.3581)    (35,8.3962)    (40,8.4087)};
    
     \addplot[blue,mark=square,every mark/.append style={solid}] plot coordinates{(0,2.4663)    (5,4.6525)    (10,7.1437)    (15,9.2133)    (20,10.4963)    (25,11.1207)    (30,11.3755)    (35,11.4682)    (40,11.4998)};
     
      \addplot[blue,mark=diamond] plot coordinates{(0,1.7265)    (5.25,3.77)    (9.75,6.81)    (15,10.75)    (19.5,15.03)   (24.85,19.93)   (30.1,24.82)   (35.4,30.04)   (39.9,34.66) };

     \addplot[blue,dashed,mark=x,every mark/.append style={solid}] plot coordinates{(0,2.0569)    (5,3.4207)    (10,4.5922)    (15,5.3169)    (20,5.6651)    (25,5.8046)    (30,5.8542)    (35,5.8707)    (40,5.8760)};
     
     \addplot[blue,dashed,mark=o,every mark/.append style={solid}] plot coordinates{(0,2.3220)    (5,4.1550)    (10,5.9851)    (15,7.3029)    (20,8.0284)    (25,8.3532)    (30,8.4785)    (35,8.5226)    (40,8.5373)};
  
     \addplot[blue,dashed,mark=square,every mark/.append style={solid}] plot coordinates{(0,2.4387)    (5,4.6090)    (10,7.0982)    (15,9.2025)    (20,10.5508)    (25,11.2377)    (30,11.5320)    (35,11.6438)    (40,11.6828)};
     
      \addplot[blue,dashed,mark=diamond] plot coordinates{(0,1.7240)    (5.25,3.7240)    (9.75,6.6640)    (15,10.8420)    (19.5,14.9883)};

 \legend{  {$B_{\rm G}=5$},{$B_{\rm G}=10$}, {$B_{\rm G}=15$}, {$B_{\rm G}=\frac{N_{\rm G}}{2} \log P$}}
    \end{axis}
  \end{tikzpicture}
  \caption{Comparison of the perturbation scheme from Section~\ref{GMperturbation} (solid) to the real quantizer \eqref{eq_quantizer} (dashed), for the $3$-user MIMO IC, $N=M=2$.}
  \label{fig:perturb}
\end{figure}
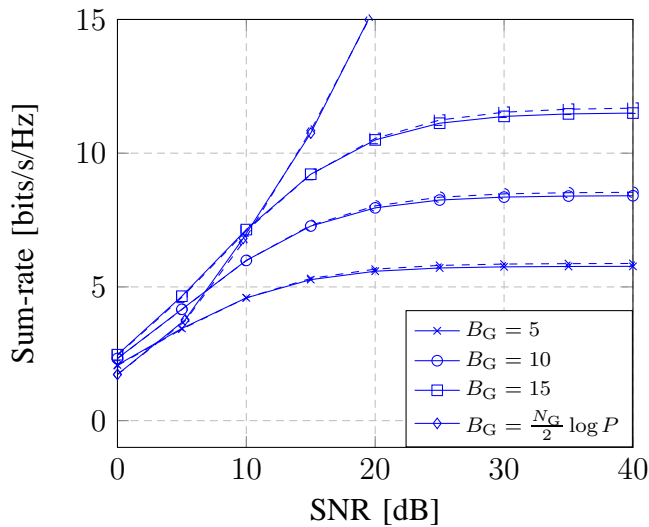
\begin{figure}
  \centering
  \begin{tikzpicture}[scale=1]
    \renewcommand{\axisdefaulttryminticks}{4}
    \pgfplotsset{every major grid/.append style={densely dashed},every mark/.append style={solid}}
    \tikzstyle{every axis y label}+=[yshift=-10pt]
    \tikzstyle{every axis x label}+=[yshift=5pt]
    \pgfplotsset{every axis legend/.append style={cells={anchor=west},fill=white, at={(0,1)}, anchor=north west,{font=\scriptsize}}}
    \begin{axis}[ 
      xmin=0,
      ymin=-5,
      xmax=40,
      ymax=40,
      bar width=3pt,
      grid=major,
      scaled ticks=true,
      ylabel={Sum-rate [bits/s/Hz]},
      xlabel={{\rm SNR} [dB]}
      ]            
   
     \addplot[blue,dashed,mark=x,every mark/.append style={solid}] plot coordinates{(0,2.4756)    (5,4.7990)    (10,7.6703)    (15,10.2614)    (20,11.8843)    (25,12.6068)    (30,12.8656)    (35,12.9498)    (40,12.9764)};
      
     \addplot[red,mark=square,every mark/.append style={solid}] plot coordinates{(0,2.5808)    (5,5.0770)    (10,8.3694)    (15,11.8417)    (20,14.7583)    (25,16.6810)    (30,17.6855)    (35,18.1190)    (40,18.2824)};

     \addplot[blue,dashed,mark=o,every mark/.append style={solid}] plot coordinates{(0,-2.3512)    (5,0.4213)    (10,3.6788)    (15,8.1389)    (20,12.4684)    (25,17.5059)    (30,22.6206)    (35,27.2511)    (40,32.4126)};
   
      \addplot[red,mark=diamond] plot coordinates{(0,1.7265)    (5.25,3.88)    (9.75,6.6530)    (15,10.75)    (19.5,14.73)   (24.85,19.79)   (30.1,24.82)   (35.4,30.19)   (39.9,34.66) };
           
       \addplot[black,smooth,mark=triangle] plot coordinates{(0,2.6228)    (5,5.2147)    (10,8.8014)   (15,13.0779)  (20,17.7419)     (25,22.5930)    (30,27.5247)    (35,32.4884)    (40,37.4642)     } ;

 \legend{  {$B_{\rm G}=25$, lower bound }, {$B_{\rm G}=25$ perturbation },{$B_{\rm G}=\frac{N_{\rm G}}{2}\log  P$, lower bound }, {$B_{\rm G}=\frac{N_{\rm G}}{2}\log  P$, perturbation},{Perfect CSI}}
    \end{axis}
  \end{tikzpicture}
  \caption{Sum rate performance using the perturbation method compared to perfect CSI and the lower bound derived in \eqref{in28}, for the $3$-user MIMO IC, $N=M=2$.}
  \label{fig:lb}
\end{figure}
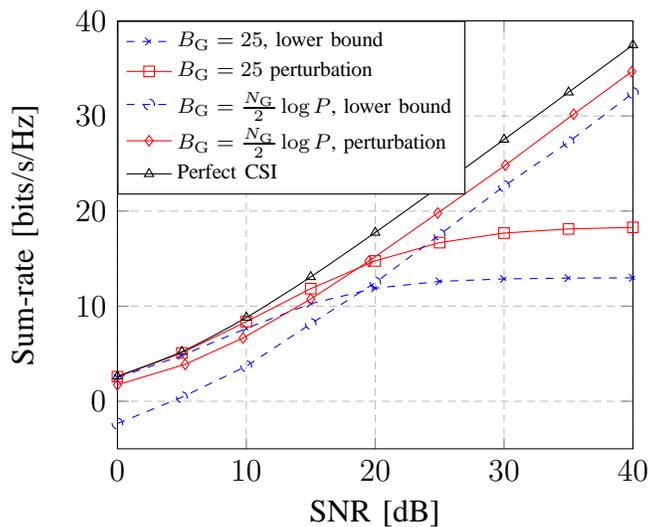

\section{Conclusion}
\label{sec:conclusion}
A new CSI feedback scheme for interference alignment on the K-user MIMO interference channel was proposed consisting in a parsimonious representation based on the Grassmann manifold. We characterized the scaling of the number of feedback bits with the SNR required in order to preserve the multiplexing gain achievable using perfect CSI. Simulations results confirm that our scheme provides a better sum rate performance compared to quantization of the normalized channel matrices for the same number of feedback bits. Furthermore, considering quantization on the Grassmann manifold, we introduced a model for the chordal distance of the quantization error which facilitates the numerical performance analysis of schemes requiring intractably large codebooks.\\

\appendices

\section{Proof of Lemma~\ref{thm_boundedleakage}}
\label{appendix_leakage_bound}

The power of the interference leakage at receiver $i$ reads
\begin{equation}\begin{split}\label{e9}
L_i&={\rm tr}\left(\frac{P}{d}\sum\limits_{j= 1, j\neq i}^K {\bf G}_i^{\rm H}{{{{\bf H}}}_{ij}}{{{{\bf V}}}_j}{{{{\bf V}}}_j^{\rm H}}{{{{\bf H}}}_{ij}^{\rm H}}{\bf G}_i\right)\\
&=\frac{P}{d}{\rm tr}\left({\bf G}_i^{\rm H}{\bf H}_i{\bf V}_{-i}{\bf V}_{-i}^{\rm H}{\bf H}_i^{\rm H}{\bf G}_i\right)\\
&= \frac{P}{d}||{\bf G}_i^{\rm H}{\bf H}_i{\bf V}_{-i} ||_{\rm F}^2.\\
\end{split}\end{equation}
Substituting ${\bf G}_i^{\rm H}=({\bf C}_i^{-1}{\bf F}_i^{\rm H}{\hat {\bf F}}_i{\tilde {\bf U}}_i)^{\rm H}$ and ${\bf H}_i={\bf C}_i^{\rm H}{\bf F}_i^{\rm H}$ gives
\begin{equation}\begin{split}\label{e99}
L_i&=\frac{P}{d}||{\tilde {\bf U}}_i^{\rm H}{\hat {\bf F}}_i^{\rm H}{\bf F}_i{\bf C}_i^{-{\rm H}}{\bf C}_i^{\rm H}{\bf F}_i^{\rm H}{\bf V}_{-i} ||_{\rm F}^2\\
&=\frac{P}{d}||{\tilde {\bf U}}_i^{\rm H}{\hat {\bf F}}_i^{\rm H}{\bf F}_i{\bf F}_i^{\rm H}{\bf V}_{-i} ||_{\rm F}^2.\\
\end{split}\end{equation}
Using the alignment equation \eqref{e8} and the fact that  ${\hat {\bf F}}_i^{\rm H}{\hat {\bf F}}_i={\bf I}_N$ yields ${\tilde {\bf U}}_i^{\rm H}{\hat {\bf F}}_i^{\rm H}{\hat {\bf F}}_i{\hat {\bf F}}_i^{\rm H}{\bf V}_{-i}=0$, therefore \eqref{e99} can be rewritten as
\begin{equation}\begin{split}\label{e299}
L_i=\frac{P}{d}||{\tilde {\bf U}}_i^{\rm H}{\hat {\bf F}}_i^{\rm H}({\bf F}_i{\bf F}_i^{\rm H}-{\hat {\bf F}}_i{\hat {\bf F}}_i^{\rm H}){\bf V}_{-i} ||_{\rm F}^2.\\
\end{split}\end{equation}  
Using the facts that $||{\bf X}||_{\rm F} \leq \sqrt{{\rm rank}({\bf X})} \, ||{\bf X}||_2$, $||{\bf X}||_2 \leq ||{\bf X}||_{\rm F}$ and $||{\bf XY}||_2\leq ||{\bf X}||_2 \, ||{\bf Y}||_2$, we have
\begin{equation}\begin{split}\label{e288}
L_i&=\frac{P}{d}||{\tilde {\bf U}}_i^{\rm H}{\hat {\bf F}}_i^{\rm H}({\bf F}_i{\bf F}_i^{\rm H}-{\hat {\bf F}}_i{\hat {\bf F}}_i^{\rm H}){\bf V}_{-i} ||_{\rm F}^2\\
&\leq P||{\tilde {\bf U}}_i^{\rm H}{\hat {\bf F}}_i^{\rm H}({\bf F}_i{\bf F}_i^{\rm H}-{\hat {\bf F}}_i{\hat {\bf F}}_i^{\rm H}){\bf V}_{-i} ||_2^2\\
&\leq P||{\tilde {\bf U}}_i^{\rm H}||_2^2 \, ||{\hat {\bf F}}_i^{\rm H}||_2^2 \, ||({\bf F}_i{\bf F}_i^{\rm H}-{\hat {\bf F}}_i{\hat {\bf F}}_i^{\rm H})||_2^2 \, ||{\bf V}_{-i} ||_2^2\\
&= P||({\bf F}_i{\bf F}_i^{\rm H}-{\hat {\bf F}}_i{\hat {\bf F}}_i^{\rm H})||_2^2\\
&\leq P||({\bf F}_i{\bf F}_i^{\rm H}-{\hat {\bf F}}_i{\hat {\bf F}}_i^{\rm H})||_{\rm F}^2\\
&= 2Pd_c^2({\hat {\bf F}}_i,{\bf F}_i).
\end{split}\end{equation}
The second equality holds because ${\tilde {\bf U}}_i^{\rm H}$, ${\hat {\bf F}}_i^{\rm H}$ and ${\bf V}_{-i} $ are truncated unitary matrices, which implies that their spectral norm is 1.
%

From \cite[Theorem 5]{rajathesis}, if a codebook is generated using the sphere-packing procedure, the maximum value of the quantization error in terms of the chordal distance can be upper bounded as
\begin{equation}\label{bound_delta_imax}
\max_{{\bf F}_i \in {\mathcal G}_{(K-1)M,N}} d_c({\hat {\bf F}}_i,{\bf F}_i) \leq \frac{2}{{(c\, 2^{B_{\rm G}})}^{\frac{1}{N_{\rm G}}}} \left(1+ o\left(2^{-\frac{B_{\rm G}}{N_{\rm G}}}\right)\right).
\end{equation}
The constant $c$ in \eqref{cequ} is obtained from \cite[Corollary 1]{dai}.
Combining \eqref{e288} and \eqref{bound_delta_imax} yields \eqref{e293}. \\

\section{Proof of Lemma \ref{inlem2}}
\label{appendix_scaling_naive}
Similar to \eqref{e9}, the power of the interference leakage at receiver $i$ can be written as
\begin{align}\label{err9}
{\bar L}_i&={\rm tr}\left(\frac{P}{d}\sum\limits_{j= 1, j\neq i}^K {\tilde {\bf U}}_i^{\rm H}{{{{\bf H}}}_{ij}}{{{{\bf {\bf V}}}}_j}{{{{\bf {\bf V}}}}_j^{\rm H}}{{{{\bf H}}}_{ij}^{\rm H}}{\tilde {\bf U}}_i\right)\\
&= \frac{P}{d}\sum\limits_{j= 1, j\neq i}^K ||{\tilde {\bf U}}_i^{\rm H}{\bf H}_{ij}{\bf V}_{j} ||_{\rm F}^2\\
&= \frac{P}{d}\sum\limits_{j= 1, j\neq i}^K ||{\tilde {\bf U}}_i^{\rm H}({\bf H}_{ij}-\alpha{\bf \hat H}_{ij}){\bf V}_{j} ||_{\rm F}^2\\
&\leq \frac{P}{d}\sum\limits_{j= 1, j\neq i}^K ||{\tilde {\bf U}}_i^{\rm H}||_{\rm F}^2 \ ||({\bf H}_{ij}-\alpha{\bf \hat H}_{ij})||_{\rm F}^2 \ ||{\bf V}_{j} ||_{\rm F}^2\\
&\leq Pd\sum\limits_{j= 1, j\neq i}^K ||{\bf H}_{ij}-\alpha{\bf \hat H}_{ij}||_{\rm F}^2\\
&= Pd\sum\limits_{j= 1, j\neq i}^K ||{\rm vec}({\bf H}_{ij})-\alpha{\rm vec}({\bf \hat H}_{ij})||_2^2\\
&\leq Pd\sum\limits_{j= 1, j\neq i}^K {||{\rm vec}({\bf H}_{ij})||_2^2} \ \left|\left|{\bf z}_{ij}-\alpha \frac{{\bf \hat z}_{ij}}{||{\rm vec}({\bf H}_{ij})||_2}\right|\right|_2^2
\end{align}
for an arbitrary scalar $\alpha$. In particular, choosing $\alpha={\bf \hat z}_{ij}^{\rm H}{\bf z}_{ij}{||{\rm vec}({\bf H}_{ij})||_2}$ yields
\begin{align}\label{err15}
{\bar L}_i &\leq  Pd\sum\limits_{j=1, j\neq i}^K{||{\rm vec}({\bf H}_{ij})||_2^2} \ ||{\bf z}_{ij}{\bf z}^{\rm H}_{ij}{\bf z}_{ij}-{\bf \hat z}_{ij}{\bf \hat z}^{\rm H}_{ij}{\bf z}_{ij}||_2^2\\
&\leq  Pd\sum\limits_{j=1, j\neq i}^K{||{\rm vec}({\bf H}_{ij})||_2^2} \ ||{\bf z}_{ij}{\bf z}^{\rm H}_{ij}-{\bf \hat z}_{ij}{\bf \hat z}^{\rm H}_{ij}||_{\rm F}^2 \ ||{\bf z}_{ij}||_2^2\\
&={2Pd}\sum\limits_{j=1, j\neq i}^K{||{\rm vec}({\bf H}_{ij})||_2^2} \ (1- |{{\bf z}_{ij}^{\rm H}{\bf \hat z}_{ij}}|^2)\\
&\leq{2PdB_{\rm max}}\sum\limits_{j=1, j\neq i}^K(1- |{{\bf z}_{ij}^{\rm H}{\bf \hat z}_{ij}}|^2)\\
&={2PdB_{\rm max}}D^2_c({\bf Z}_i,{\bf \hat Z}_i)
\end{align}
where $B_{\rm max}=\max_{j}||{\rm vec}({\bf H}_{ij})||_2^2$. \\

From \cite[theorem II.1]{rajac}, the distance on the composite Grassmann manifold can be bounded for any codebook obtained via sphere-packing as $\max_{{\bf Z}_{i} \in {\mathcal G}^{K-1}_{MN,1}} D_c({{\bf  Z}}_{i},{\bf \hat Z}_{i})\leq \frac{2}{{({\bar c}\, 2^{B_{\rm INM}})}^{\frac{1}{N_{\rm INM}}}} \triangleq {\bar \Delta}$ which results in
 \begin{equation}\label{err20}
{\bar L}_i \leq {2Pd}B_{\rm max}{\bar \Delta}^2 \leq \frac{8PdB_{\rm max}}{{({\bar c}\, 2^{B_{\rm INM}})}^{\frac{2}{N_{\rm INM}}}}
 \end{equation}
where $N_{\rm INM}$ was defined in Section~\ref{section_feedbackdimension} and ${\bar c}$ is a constant.
It is clear from \eqref{err20} that quantizing ${\bf Z}_i$ with $B_{\rm INM}=\frac{N_{\rm INM}}{2}{\rm log} P$ bits at receiver $i$ guarantees that ${\bar L}_i$ remains bounded regardless of the SNR.

\section{Proof of Lemma \ref{lemma_LB_rate}}
\label{appendix_proof_LB_rate}
Consider the following quantity
\begin{eqnarray}
R_i''-R_i'&=&  \log  \left| {{\bf G}_i^{\rm H}{\bf G}_i + \frac{P}{d}{\bf Q}_{\rm I}^i} \right|-A^i\\
&\leq& \log \left| {{\bf G}_i^{\rm H}{\bf G}_i + \frac{P}{d}{\bf Q}_{\rm I}^i} \right|, \label{rl1}
\end{eqnarray}
where $A^i= \log  \left| {{\bf G}_i^{\rm H}{\bf G}_i + \frac{P}{d}({\bf Q}_{\rm S}^i+{\bf Q}_{\rm I}^i)} \right|-\log \left| {{\bf G}_i^{\rm H}{\bf G}_i + \frac{P}{d}{\bf Q}_{\rm S}^i} \right|$ and $A^i\geq 0$ since ${\bf Q}_{\rm I}^i$ is positive semi-definite\footnote{In fact, when the number of feedback bits is scaled according to \eqref{eqa},
the bound in \eqref{rl1} gets tighter as the SNR increases. This can be seen by noticing that
$\lim_{P \rightarrow \infty}A^i=\lim_{P \rightarrow \infty} \log  \left| {{{\bf I}_{{d}}} + \frac{P}{d}{\bf Q}_{\rm I}^i}\left({{\bf G}_i^{\rm H}{\bf G}_i + \frac{P}{d}{\bf Q}_{\rm S}^i}\right)^{-1} \right|=\lim_{P \rightarrow \infty} \log  \left| {{{\bf I}_{{d}}} + {\bf Q}_{\rm I}^i} {{\bf Q}_{\rm S}^i}^{-1} \right|$, which goes to zero since ${\bf Q}_{\rm S}^i$ is full rank almost surely and when feedback scales according to \eqref{eqa} we have $||{\bf Q}_{\rm I}^i||_2 \rightarrow 0$.}. 

Since the argument of the determinant is of rank at most $d$, 
\begin{eqnarray}
R_i''-R_i'& \leq& d \log \left( \lambda_{\rm max}\left({{\bf G}_i^{\rm H}{\bf G}_i + \frac{P}{d}{\bf Q}_{\rm I}^i} \right) \right)\\
& \leq& d \log \left(\lambda_{\rm max}\left( {\bf G}_i^{\rm H}{\bf G}_i \right) + \frac{P}{d} \lambda_{\rm max}\left({{\bf Q}_{\rm I}^i} \right) \right)\\
& =& d \log \left( \left|\left|{\bf G}_i\right|\right|_2^2+ \frac{P}{d}  \left|\left|{\bf G}_i^{\rm H}{\bf H}_i{\bf V}_{-i} \right|\right|_2^2\right), 
\end{eqnarray}
where the second inequality follows by the fact that ${\bf G}_i^{\rm H}{\bf G}_i $ and ${\bf Q}_{\rm I}^i$ are Hermitian matrices. Furthermore $\left|\left|{\bf G}_i\right|\right|_2 = ||{\bf C}_i^{-1}{\bf F}_i^{\rm H}{\hat {\bf F}}_i{\tilde {\bf U}}_i ||_2 \leq ||{\bf C}_i^{-1}||_2 \ ||{\bf F}_i^{\rm H} ||_2 \ ||{\hat {\bf F}}_i ||_2 \ ||{\tilde {\bf U}}_i ||_2 = ||{\bf C}_i^{-1}||_2$.
From eqs.~\eqref{e9}--\eqref{e288} we have $||{\bf G}_i^{\rm H}{\bf H}_i{\bf V}_{-i} ||_2^2 \leq 2d_c^2({\hat {\bf F}}_i,{\bf F}_i)$. Using these bounds,
\begin{eqnarray} \label{eq_log_bounded_chordal}
R_i''-R_i'& \leq& d \log \left( ||{\bf C}_i^{-1}||_2^2+ \frac{2P}{d}  d_c^2({\hat {\bf F}}_i,{\bf F}_i)\right).
\end{eqnarray}
Combining with \eqref{bound_delta_imax} yields \eqref{in24}.\\

\section{Proof of Lemma \ref{lem_dof_no_intf}}
\label{appendix_proof_dof_no_intf}

It suffices to prove that $\lim_{P\rightarrow \infty} \frac{\log \left| {{\bf G}_i^{\rm H}{\bf G}_i + \frac{P}{d}{\bf Q}_{\rm S}^i} \right|}{\log P}=d$ for almost all channel realizations.
Note however that the proof is complicated by the fact that we need to consider quantization codebooks of increasing sizes when letting $P\rightarrow \infty$; since ${\bf Q}_{\rm S}^{i}={\tilde {\bf U}}_{i}^{\rm H}{\hat {\bf F}}_{i}^{\rm H}{\bf F}_i{\bf C}_i^{\rm -H}{\bf H}_{ii}{\bf V}_{i}{\bf V}_{i}^{\rm H}{\bf H}_{ii}^{\rm H}{\bf C}_i^{-1}{\bf F}_i^{\rm H}{\hat {\bf F}}_{i}{\tilde {\bf U}}_{i}$, where $\tilde {\bf U}_i$, ${\hat {\bf F}}_i$ and ${\bf V}_i$ all depend on the choice of the codebook, it is not clear whether ${\bf Q}_{\rm S}^{i}$ admits a limit for asymptotically large SNR\footnote{Although it is clear that the subspace spanned by $\hat {\bf F}_{i}$ admits a limit on the Grassmann manifold when $B_{\rm G}\rightarrow \infty$, the definition of ${\tilde {\bf U}}_{i}$ and ${\bf V}_{i}$ as one (possibly among several) solution of \eqref{e8} prevents the extension of the convergence result to those variables.}. Therefore, we resort to compactness arguments to show that there exists a series of codebooks of increasing size for which ${\bf Q}_{\rm S}^{i}$ admits a limit.

Let us consider an infinite sequence of SNRs ${\mathcal P}=\{{P_n}\}_{n\in \mathbb{N}}$ such that $\lim_{n\rightarrow \infty} {P_n} = \infty$, as well as an infinite sequence of quantization codebooks $\{\mathcal{S}_n \}_{n\in \mathbb{N}}$, such that $|\mathcal{S}_n |=P_n^{N((K-1)M-N)}$, following \eqref{eqa}.
For each SNR value $P_n$, we let ${\hat {\bf F}}_{i,n}=\mathrm{arg} \min_{{\bf S} \in \mathcal{S}_n} \,\,\,d_{c}({\bf S},{\bf F}_i)$
 and denote $({\bf V}_{1,n},\ldots,{\bf V}_{K,n},{\tilde{\bf U}}_{1,n},\ldots,{\tilde{\bf U}}_{K,n}) \in {\mathcal G}_{M,d}^{K}\times {\mathcal G}_{N,d}^{K}$ a set of matrices constituting an IA solution based on ${\hat {\bf F}}_{i,n}$. In other words, we solve \eqref{eq_quantizer} and \eqref{e8} for each $n$, yielding an infinite series of solutions.
Let us denote ${\bf W}_n=({\hat {\bf F}}_{1,n},\ldots, {\hat {\bf F}}_{K,n},{\bf V}_{1,n},\ldots,{\bf V}_{K,n},{\tilde{\bf U}}_{1,n},\ldots,{\tilde{\bf U}}_{K,n})$. ${\mathcal G}_{(K-1)M,N}^K\times{\mathcal G}_{M,d}^{K}\times {\mathcal G}_{N,d}^{K}$ is compact, as a Cartesian product of compact sets. Therefore, we can extract a convergent subseries\footnote{In order to obtain the same convergence properties for a point on ${\mathcal G}_{a,b}$ and for the corresponding unitary matrix representation ${\bf F}\in\mathbb{C}^{a,b}$, it is useful to make this representation unique, e.g. by requiring that the top square $b\times b$ subblock of ${\bf F}$ is equal to ${\bf I}_b$. For the sake of notational simplicity, we omit those details.} from $\{ {\bf W}_n \}_{n\in\mathbb{N}}$. We let $g(m)\in\mathbb{N}$ denote the index of the $m$-th element of the convergent subseries, where $g$ is a monotonically increasing function.
We also denote
\begin{equation}\begin{split}\label{dubyastar}
({\hat {\bf F}}_1^{\star},\ldots, {\hat {\bf F}}_K^{\star},{\bf V}_{1}^{\star},\ldots,{\bf V}_{K}^{\star},{\tilde{\bf U}}_{1}^{\star},\ldots,{\tilde{\bf U}}_{K}^{\star}) = \lim_{m \rightarrow \infty} {\bf W}_{g(m)}.
 \end{split}\end{equation}
Letting ${\bf Q}_{\rm S}^{i,n}={\tilde {\bf U}}_{i,n}^{\rm H}{\hat {\bf F}}_{i,n}^{\rm H}{\bf F}_{i}{\bf C}_i^{\rm -H}{\bf H}_{ii}{\bf V}_{i,n}{\bf V}_{i,n}^{\rm H}{\bf H}_{ii}^{\rm H}{\bf C}_i^{-1}{\bf F}_{i}^{\rm H}{\hat {\bf F}}_{i,n}{\tilde {\bf U}}_{i,n}$, we can now write the limit
$\lim_{m \rightarrow \infty} {\bf Q}_{\rm S}^{i,g(m)} = {\bf Q}_{\rm S}^{\star i}$, where
${\bf Q}_{\rm S}^{\star i}={\tilde {\bf U}}_{i}^{\star \rm H}{\hat {\bf F}}_{i}^{\star \rm H}{\bf F}_i{\bf C}_i^{\rm -H}{\bf H}_{ii}{\bf V}_{i}^{\star}{\bf V}_{i}^{\star \rm H}{\bf H}_{ii}^{\rm H}{\bf C}_i^{-1}{\bf F}_i^{\rm H}{\hat {\bf F}}_{i}^{\star}{\tilde {\bf U}}_{i}^{\star}$. Therefore we have
\begin{align}
\lim_{m \rightarrow \infty} \frac{\log \left|{\bf G}_i^{\rm H}{\bf G}_i+\frac{P_{g(m)}}{d} {\bf Q}_{\rm S}^{i,{g(m)}}\right|}{\log P_{g(m)}} &=\lim_{m \rightarrow \infty} \frac{\log |\frac{P_{g(m)}}{d} {\bf Q}_{\rm S}^{\star i}|}{\log P_{g(m)}}\\
&=\mathrm{rank}\left({\bf Q}_{\rm S}^{\star i}\right).
\end{align}
Since ${\hat {\bf F}}_{i}^\star$ and ${\bf F}_i$ span the same subspace, ${\hat {\bf F}}_{i}^{\star \rm H}{\bf F}_i$ is unitary. Therefore, considering the product of matrices in ${\bf Q}_{\rm S}^{\star i}$, we note that ${\tilde {\bf U}}_{i}^{\star \rm H}{\hat {\bf F}}_{i}^{\star \rm H}{\bf F}_i{\bf C}_i^{\rm -H}$ has full row rank $d$, ${\bf V}_{i}^{\star}$ has full column rank $d$, and both are independent of ${\bf H}_{ii}$, from which we conclude that $\mathrm{rank}\left({\bf Q}_{\rm S}^{\star i}\right)=d$ a.s., which proves the lemma.\\

\section{Proof of Theorem~\ref{thm_rvq}}

\label{appendix_proof_rvq}
Let us first recall that $R_i\geq R_i'$, which holds also in expectation:
\begin{equation} \label{expect_Ri_Rip}
{\rm E}_{{\mathcal S}}(R_i)\geq {\rm E}_{{\mathcal S}}(R_i').
\end{equation}
Furthermore, from \eqref{eq_log_bounded_chordal}, 
\begin{equation}\begin{split}\label{in25}
{\rm E}_{{\mathcal S}}(R_i')
& \geq R_i''-d \ {\rm E}_{{\mathcal S}}\left(\log \left(||{\bf C}_i^{-1}||_2^2+ \frac{2P}{d}  d_c^2({\hat {\bf F}}_i,{\bf F}_i)\right)\right) \\
& \geq R_i''-d \log \left(||{\bf C}_i^{-1}||_2^2+ \frac{2P}{d}  {\rm E}_{{\mathcal S}}\left(d_c^2({\hat {\bf F}}_i,{\bf F}_i)\right)\right)
\end{split}\end{equation}
where the second inequality follows by application of Jensen's inequality to the $\log$ function. The term ${\rm E}_{{\mathcal S}}(d_c^2({\hat {\bf F}}_i,{\bf F}_i))$ represents the expected value of the distortion while using a random codebook, and can be further bounded using \cite[Theorem~6]{rajathesis}, which can be summarized as follows:
for asymptotically large codebook size, when using a random codebook for quantizing a matrix ${\bf F}$ arbitrarily distributed over a manifold, the $k$-th moment of the chordal distance $D^{(k)}={\rm E}_{{\mathcal S},{\bf F}}(d_c^k({\hat {\bf F}},{\bf F}))$ can be bounded as
\begin{equation}\label{in27}
\frac{N_m}{{{(N_m+k)}}{(c\, 2^{B_{\rm G}})}^{\frac{k}{N_m}}} \leq D^{(k)} \leq   \frac{\Gamma ({\frac{k}{N_m}})}{{\frac{N_m}{k}}{(c\, 2^{B_{\rm G}})}^{\frac{k}{N_m}}},
\end{equation}
where the codebooks have $2^{B_{\rm G}}$ elements and $N_m$ is the real dimension of the corresponding manifold.  
Using the upper bound in \eqref{in27} for $k=2$ over the Grassmann manifold, combined with \eqref{expect_Ri_Rip} and \eqref{in25} results in \eqref{in28}.\\

\section*{Acknowledgments}
{
The authors would like to thank Omar El~Ayach from UT~Austin for his helpful comments. This work was supported by the FP7 project HIATUS (grant $\#$265578) of the European Commission and by the Austrian Science Fund (FWF) through grant NFN SISE (S106). We also acknowledge the support of the Newcom\# Network of Excellence in Wireless Communications of the EC.

\bibliographystyle{IEEEtran}
%

\balance
\bibliography{biblio-maxime,mybib2}
}

\end{document}